%% file: main.tex
\documentclass{article}
\include{preamblesAndMacros}

\usepackage[margin=1.25in]{geometry} 
\title{Johnson-Lindenstrauss embedding for noisy vectors -- taking advantage of the noise}
\author{Zhen Shao 
}
\date{1 September 2022}

\begin{document}

\maketitle

\begin{abstract}
    This paper investigates theoretical properties of subsampling and hashing as tools for approximate Euclidean norm-preserving embeddings for vectors with (unknown) additive Gaussian noises. Such embeddings are sometimes called Johnson-lindenstrauss embeddings due to their celebrated lemma. Previous work shows that as sparse embeddings, the success of subsampling and hashing closely depends on the $l_\infty$ to $l_2$ ratios of the vector to be mapped. This paper shows that the presence of noise removes such constrain in high-dimensions, in other words, sparse embeddings such as subsampling and hashing with comparable embedding dimensions to dense embeddings have similar approximate norm-preserving dimensionality-reduction properties. The key is that the noise should be treated as an information to be exploited, not simply something to be removed. Theoretical bounds for subsampling and hashing to recover the approximate norm of a high dimension vector in the presence of noise are derived, with numerical illustrations showing better performances are achieved in the presence of noise. 
\end{abstract}

\section{Introduction}

Dimensionality reduction that approximately preserves Euclidean distances is a key tool used in many algorithms \cite{MR1715608, BCGNPaper}. A fundamental result is the Johnson-Lindenstrauss Lemma.

\begin{lemma}[JL Lemma \cite{Johnson:1984aa, MR1943859}]
\label{def::JL_lemma}
Given a fixed, finite set $Y\subseteq \R^{n}$, $\epsilonS, \deltaS > 0$, let $S \in\R^{m\times n}$ 
have entries independently distributed as the normal $N(0, n^{-1})$, with 
$m = \mathO{\epsilonS^{-2}\log(\frac{|Y|}{\deltaS})}$ and where $|Y|$ refers to the cardinality of the set $Y$. Then we have, with probability at least $1-\deltaS$, that
\begin{equation}
(1-\epsilonS)\|y\|_2^2 \leq \|Sy\|_2^2 \leq (1+\epsilonS)\|y\|_2^2 \quad \text{for all}\,\, y \in Y.
\end{equation}
\end{lemma}

The target dimension $m$ is known to be optimal \cite{MR3734267, MR3073520}. 

However, the density of Gaussian matrices results in relatively large computational costs, both in terms of evaluating the matrix-vector product, or in terms of evaluating the vector itself, when the entries of the vector need to be computed from expensive computational procedures.  To alleviate these deficiencies, subsampling techniques \cite{10.5555/1109557.1109682} and 
 sparse random matrix ensembles \cite{10.1145/3019134} have been proposed as alternatives for computational efficiency; the former consists of applying matrices with $1$ non-zeros per row at a column uniformly chosen at random, and the latter consists of $s$-hashing matrices with $s$ (fixed) non-zeros per column, at random rows, $s \geq 1$. Nevertheless, without data-specific assumptions, subsampling fails to give an approximate norm preserving embedding in general, and $s$-hashing can only succeed with a relatively large $s$, as shown by the lower bound in \cite{10.1145/2488608.2488622, Nelson:2014uu}.
 
 Real data, however, tend to be corrupted by random noises. Therefore the pathological example mentioned in \cite{2021arXiv210511815C}, and the lower bound in \cite{Nelson:2014uu} may not apply when we consider an alternative data model, namely, each data input $x$ comes with Gaussian noise. This framework was similar to the smoothed analysis of algorithmic complexities for the simplex algorithm \cite{MR2120328}. Specifically, we consider our inputs as
\begin{equation}
    \tilde{x}  = x + \sigma \normTwo{x} r, \label{eqn:noisyX}
\end{equation}
where $x,r \in \R^n$, $\sigma > 0$, and $r$ have i.i.d. entries distributed as $N(0,1)$.

We show, under this noisy data model, that we can not only embed noisy vector $\tilde{x}$ using subsampling or hashing matrices with significantly smaller embedding dimensions while approximately preserving $\normTwo{\tilde{x}}$ , but since $\normTwo{\tilde{x}}$ is intimately connected to $\normTwo{x}$, we can simulataneously approximately recover $\normTwo{x}$, thus removing the noise. The combination of these two insights results in the surprising fact that when noise is present, approximate norm-preserving embedding becomes more efficient, due to the fact that the noise is correlated with our target -- the norm of the vector, and therefore we take advantage of the presence of the noise. 

A major innovation in this paper consists of proposing to consider embedding properties in a noisy input regime. We derive the following main result, which shows that the very sparse $1$-hashing matrices have a comparable embedding property as Gaussian matrices in a high-dimensional, noisy input setting.
The following expression will be needed in our results,
\begin{equation}
\bar{\nu}(\epsilonS,\deltaS):= \nuOneHashingSCTwo,
      \label{L5:nu}
\end{equation}
where $\epS, \deltaS \in (0,1)$ and $E, C_2>0$.

\begin{theorem}\label{thm:hashing}
Let $n \in \N^+$, $\epsilon, \epS, \deltaS \in (0,1)$. $t, \sigma, C_1 > 0$ with $2 \sigma t < 1$.  Suppose that $n \geq n_0$, where
\begin{equation}
    n_0 = \frac
    { \bracket{\frac{\nu(x) + \sigma \sqrt{2C_1 \log(2n)}}{\nuBar}}^2 - 1 + 2\sigma t}
    {\sigma^2 \bracket{1 - \epsilon}}
    \label{eqn:n0}
\end{equation}
Then given a fixed vector $x \in \R^n$, let $S \in \R^{m\times n}$ be a $1$-hashing matrix, with $m = \mathO{\epS^{-2} \logFrac{1}{\deltaS}}$, with probability at least $\bracket{1 - 2e^{-\frac{n\epsilon^2}{12}} - 2e^{-t^2/2} - \bracket{\frac{1}{2n}}^{C_1 -1}}\bracket{1 - \deltaS}$, we have that
\begin{equation}
(1-\epS)\bracket{1 - \noisyVectorError}\|x\|_2^2 \leq \frac{\|S\tilde{x}\|_2^2}{1 + \sigma^2 n} \leq (1+\epS)\bracket{1 + \noisyVectorError}\|x\|_2^2,
\end{equation}
where $\tilde{x}$ is defined in \eqref{eqn:noisyX} and $r$ in \eqref{eqn:noisyX} is independent of $S$.

\end{theorem}

In addition, we derive results for using subsampling techniques in a noisy data setting. Subsampling a vector (with rescaling) is equivalent to applying a (scaled) sampling matrix.

\begin{definition}\label{def:sampling}
We define $S \in \R^{m \times n}$ to be a scaled sampling matrix if, independently for each $i \in [m]$, we sample $j \in [n]$ uniformly at random and let $S_{ij}=\sqrt{\frac{n}{m}}$. 
\end{definition}

\begin{theorem}\label{thm:sampling}
Let $n \in \N^+$, $\epsilon, \epS, \deltaS \in (0,1)$. $t, \sigma, C_1 > 0$ with $2 \sigma t < 1$. Then given a fixed vector $x \in \R^n$, let $S \in \R^{m\times n}$ be a scaled sampling matrix (\autoref{def:sampling}), with 
\begin{equation}
    m = 2 \epsilonS^{-2} \logFrac{1}{\deltaS} \frac{ \bracket{\nu(x) + \sigma\sqrt{2C_1 \log(2n)}}^2}{ \bracket{\sigma^2 + n^{-1}} \bracket{1 - \noisyVectorError}}, 
\end{equation}
with probability at least $\bracket{\noisyVectorOverallProb}\bracket{1 - \deltaS}$, we have that
\begin{equation}
(1-\epS)\bracket{1 - \noisyVectorError}\|x\|_2^2 \leq \frac{\|S\tilde{x}\|_2^2}{1 + \sigma^2 n} \leq (1+\epS)\bracket{1 + \noisyVectorError}\|x\|_2^2,
\end{equation}
where $\tilde{x}$ is defined in \eqref{eqn:noisyX} and $r$ in \eqref{eqn:noisyX} is independent of $S$. 
\end{theorem}

The proofs of the above results are contained in Section 2.

\subsection{Related work}

The key concept in the Johnson-Lindenstrauss Lemma is the Johnson-Lindenstrauss (JL) embedding.

\begin{definition}[JL embedding \cite{10.1561/0400000060}]
\label{def::JL_embedding}
An $\epsilon$-JL embedding for a set $Y\subseteq \R^{n}$ 
is a matrix $S \in\R^{m\times n}$ such that
\begin{equation}\label{JL}
(1-\epsilon)\|y\|_2^2 \leq \|Sy\|_2^2 \leq (1+\epsilon)\|y\|_2^2 \quad \text{for all}\,\, y \in Y.
\end{equation}
\end{definition}

Oblivious embeddings are matrix distributions such that given a(ny)  subset/column subspace of vectors in $\R^n$, a random matrix drawn 
 from such a distribution is an embedding for these vectors with high probability. We let $1-\delta\in [0,1]$ denote a(ny) success probability of an embedding.

\begin{definition}[Oblivious embedding \cite{10.1561/0400000060,10.1109/FOCS.2006.37}] \label{Oblivious_embedding}
A distribution $\cal{S}$ on $S \in \R^{m \times n}$  is an $(\epsilon,\delta)$-oblivious embedding if given a fixed/arbitrary set of vectors, we have that, with probability at least $1-\delta$, a matrix $S$ from the distribution is an $\epsilon$-embedding for these vectors.
\end{definition}

Since the celebrated Johnson and Lindenstrauss Lemma, a line of work starts on improving the computational efficiency of oblivious $JL$-embeddings. Using sparse Gaussian matrices with fixed density was proposed in \cite{MR2005771}, using a fast Fourier transform-based embedding was proposed in \cite{10.1145/1132516.1132597}. Using $1$-hashing matrices was proposed as a feature hashing technique in \cite{10.1145/1553374.1553516}, which is a special case of $s$-hashing matrices, defined as

\begin{definition}\cite{10.1145/3019134} \label{def::sampling_and_hashing}
$S \in \R^{m \times n}$ is a $s$-hashing matrix if independently for each $j \in [n]$, we sample without replacement $i_1, i_2, \dots, i_s \in [m]$ uniformly at random and let $S_{i_k j} = \pm 1/\sqrt{s}$ for $k = 1, 2, \dots, s$.
\end{definition}

It follows that when $s=1$, 	$S \in \R^{m \times n}$ is a {\it $1$-hashing matrix} if independently for each $j \in [n]$, we sample $i \in [m]$ uniformly at random and let $S_{ij} = \pm1$ \footnote{The random signs are so as  to ensure that in expectation, the matrix $S$ preserves the norm of a vector $x \in \R^n$.}. The analysis for a general $s$ was done by \cite{MR3167920, MR3773205}, and a lower bound is found in \cite{10.1145/2488608.2488622}. 

Data-dependent analysis for $1$-hashing matrices, specifically, analysis for embedding vectors with small $l_\infty$ to $l_2$ norm ratios are found in \cite{10.5555/3295222.3295407, 10.5555/3327345.3327444}. The $l_\infty$ to $l_2$ norm ratio quantifies the non-uniformity of entries.

\begin{definition}[Non-uniformity of a vector]\label{def:non-uniform-vector}
Given $x \in \R^{n}$, the non-uniformity of $x$, $\nu(x)$, is defined as
\begin{align}
\nu(x) = \frac{\|x\|_{\infty}}{\|x\|_2}. 
\end{align}
\end{definition}

\cite{10.5555/3327345.3327444} shows that $1$-hashing matrices is an oblivious $JL$-embedding for vectors with sufficiently small $\nu$. Data-dependent analysis for $s$-hashing matrices is recently done in \cite{NIPS2019_9656}.

The smoothed complexity idea was proposed by \cite{MR2120328} in the context of the simplex algorithm, a long-standing puzzle for which is its exponential worst-case complexity and its excellent practical performance. The authors in \cite{MR2120328} defines 
the smoothed complexity of the simplex algorithm as
\begin{equation}
    C_{\text{smooth}}(n, \sigma) = \max_x \mathbb{E}_{r\in X_n} T(x + \sigma \normTwo{x} r),
\end{equation}
where $X_n$ is a chosen distribution for the input to the simplex algorithm $x$, and $\sigma$ is a typically small constant modeling the noise in the actual input. Smoothed analysis has been applied extensively in algorithmic analysis and computer science, see \cite{arthur2011smoothed} for application to the $k$-means method, \cite{bhaskara2014smoothed} for application to tensor decompositions, and \cite{etscheid2017smoothed} for application to the local search method.

In this work, we apply the noisy model of smoothed analysis to JL-embeddings, and show $1$-hashing matrices and scaled sampling matrices have comparable dimensionality reduction properties to the Gaussian matrices in this setting, while being more computationally efficient.

\section{Proof of the main result}

We first state some familiar results in high-dimensional probability. For more details, see \cite{MR3837109}.

\paragraph{The maximum of a sequence of independent Gaussian}

The first lemma we will use concerns the maximum of a sequence of independent standard normal random variables.

\begin{lemma}\label{lem:maxGaussBound}
Let $Z_n = \max\{|X_1|, |X_2|, \dots, |X_n|\}$ where $X_i$ are $N(0,1)$, then we have, for $t>0$, 

\begin{equation}
    \probability{Z_n \geq t} \leq 2n e^{-t^2/2}. \label{eqn:maxGauss}
\end{equation}

So for any constant $C_1$, by choosing $t = \sqrt{2C_1\sigma_1 \log(2n)}$, we have $\probability{Z_n \geq t} \leq \bracket{\frac{1}{2n}}^{C_1-1}$.
\end{lemma}

\begin{proof}
Let $X$ be $N(0,1)$, then for $t > 0$, the well-known Gaussian tail bound gives 
\begin{equation*}
    \probability{X \geq t} \leq e^{-t^2/2}.
\end{equation*}
Therefore by symmetry, we have that
\begin{equation}
    \probability{\abs{X} \geq t} \leq 2e^{-t^2/2}. 
\end{equation}
Taking a union bound gives \eqref{eqn:maxGauss}.
\end{proof}

\paragraph{Concentration bound for the 2-norm of a Gaussian vector}
The next standard lemma bounds the norm of a high-dimensional Gaussian vector.

\begin{lemma} \label{lem:twoNormGaussBounnd}
Let $r_n \in \R^n$ with each entry i.i.d. $N(0,1)$, using a Chernoff-type bound, we have (see \cite{ZhenThesis}, (4.4.13), or \cite{MR3837109}), 
\begin{equation}
    \probability{ n (1-\epsilon) \leq \normTwo{r_n}^2 \leq 
    n (1+\epsilon)} \geq 1 - 2e^{-\frac{n\epsilon^2}{12}},
\end{equation}
for any $\epsilon \in (0,1)$.\footnote{The constant $12$ is not optimal, but a different constant here does not affect the conclusions of this paper.}
\end{lemma}

\subsection{The relationship between the norm of a noisy vector and the norm of the original vector}

In this subsection, we show that the norm of a noisy vector enables us to approximately compute the norm of the origianl vector. 

\begin{lemma}\label{lem:noisyvecNorm}
Let $\epsilon \in (0,1), t>0$. Let $x \in \R^n$, $\tilde{x}$ be defined in \eqref{eqn:noisyX}. Suppose that $2\sigma t < 1$. Then with probability at least $1 - 2e^{-n\epsilon^2/12} - 2e^{-t^2/2}$, we have that
\begin{equation}
    (1 + \sigma^2 n) \bracket{ 1 - \frac{\epsilon \sigma^2 n + 2 \sigma t}{1 + \sigma^2 n}} \normTwo{x}^2
    \leq 
    \normTwo{\tilde{x}}^2
    \leq 
    (1 + \sigma^2 n) \bracket{ 1 + \frac{\epsilon \sigma^2 n + 2 \sigma t}{1 + \sigma^2 n}} \normTwo{x}^2. \label{eq:tmp1}
\end{equation}
\end{lemma}

\begin{proof}
We have that 
\begin{equation}
    \normTwo{\tilde{x}}^2 = \normTwo{x}^2 + \sigma^2 \normTwo{x}^2 \normTwo{r}^2 + 2\sigma \normTwo{x}^2 \innerProduct{\frac{x}{\normTwo{x}}}{r},
\end{equation}
where $r$ has i.i.d. $N(0,1)$ entries and $\innerProduct{\frac{x}{\normTwo{x}}}{r}$ is distributed as $N(0,1)$. Therefore, applying \autoref{lem:maxGaussBound} (with $n=1$) and \autoref{lem:twoNormGaussBounnd} with the union bound gives the conclusion. 
\end{proof}

\begin{remark}
$2\sigma t < 1$ typically holds because $\sigma$ is typically much less than $1$ while $t$ can be taken as an $\mathO{1}$ constant.
\end{remark}

\subsection{A bound of non-uniformity $\nu$ for a noisy high-dimensional vector}

In this subsection, we state and prove a key result concerning the non-uniformity $\nu$ of a noisy vector $\tilde{x}$.

\begin{lemma}\label{lem:lInftoltwo}
Let $x \in \R^n$, and $\tilde{x}$ be defined in \eqref{eqn:noisyX}. Then with probability at least $1 - 2e^{-\frac{n\epsilon^2}{12}} - 2e^{-t^2/2} - \bracket{\frac{1}{2n}}^{C_1 -1}$, we have that simultaneously, \eqref{eq:tmp1} holds and 
\begin{equation}
    \nu(\tilde{x}) \leq \frac{\nu(x) + \sigma \sqrt{2C_1 \log(2n) }}{ \sqrt{ \bracket{1 + \sigma^2 n} \bracket{ 1 - \frac{\epsilon \sigma^2 n + 2 \sigma t}{1 + \sigma^2 n} }}}, \label{eq:nuxHashingUpper}
\end{equation}  
where $\nu$ is defined in \autoref{def:non-uniform-vector}.

\end{lemma}

\begin{proof}
We have that, with $x \in \R^n$, $r$ be a $n$-dimensional Gaussian vector with entries independent $N(0,1)$, $\sigma > 0$, 

\begin{align*}
    \nu(\tilde{x})^2 & = \frac{ \normInf{x + \sigma \normTwo{x} r}^2 }{ \normTwo{\tilde{x}}^2} \\
                & \leq \frac{\bracket{\normInf{x} + \sigma \normTwo{x}\normInf{r}}^2}{\normTwo{\tilde{x}}^2} \\
                & \leq \frac{ \bracket{\nu(x) + \sigma \normInf{r}}^2}{ \normTwo{\tilde{x}}^2 / \normTwo{x}^2}.
\end{align*}
Applying \autoref{lem:noisyvecNorm} and \autoref{lem:maxGaussBound} with the union bound gives the desired result. 
\end{proof}

\subsection{Proof of \autoref{thm:hashing}}
The following technical lemma by Frenksen, Kamma, and Larsen is crucial to our result. Its proof can be found in their paper \cite{10.5555/3327345.3327444}.

\begin{lemma}[\cite{10.5555/3327345.3327444}, Theorem 2] \label{Freksen}
Suppose that $\epsilonS, \deltaS \in (0,1)$, and $E$ satisfies $C \leq E < \frac{2}{\deltaS \log(1/\deltaS)}$, where $C>0$ and $C_1$ are  problem-independent constants. Let 
$m \leq n \in \N$ with
$m\geq E \epsilonS^{-2} \log(1/\deltaS)$.

Then,  for any $x\in \R^n$ with 
\begin{equation}
    \nu(x) \leq\bar{\nu}(\epsilonS,\deltaS),
    \label{L5:nu-1}
\end{equation}
where $\bar{\nu}(\epsilonS,\deltaS)$ is defined in \eqref{L5:nu},
a randomly generated 1-hashing matrix $S\in \R^{m\times n}$ is an $\epsilonS$-JL embedding for $\{x\}$ with probability at least $1-\deltaS$.

\end{lemma}

We are ready to prove \autoref{thm:hashing}.

\begin{proof}[Proof of \autoref{thm:hashing}]
Let $x \in \R^n$.
Using \autoref{lem:lInftoltwo}, $n \geq n_0$, we have that with probability at least $\noisyVectorOverallProb$, $\nu(\tilde{x}) \leq \nuBar$, where $\nuBar$ is defined in \eqref{L5:nu}, while simulatneoulsy \eqref{eq:tmp1} holds. 
Conditioning on $\nu(\tilde{x}) \leq \nuBar$, \autoref{Freksen} gives that with probability at least $1-\delta$, a randomly generated 1-hashing matrix $S\in \R^{m\times n}$ is an $\epsilonS$-JL embedding for $\{\tilde{x}\}$.
Since the randomness in $x$ and the randomness in $S$ are independent, we have that with a probability at least 
$\bracket{\noisyVectorOverallProb} \bracket{ 1 - \deltaS}$, a randomly generated 1-hashing matrix $S\in \R^{m\times n}$ is an $\epsilonS$-JL embedding for $\{\tilde{x}\}$ while \eqref{eq:tmp1} holds. Applying \eqref{eq:tmp1} gives the embedding result for $\{ x\}$.
\end{proof}

\subsection{Proof of \autoref{thm:sampling}}
Similar to the proof of the hashing main result, we make use of a non-uniformity-dependent result on subsampling.

\begin{lemma}[\cite{BCGNPaper}]\label{lem:BCGNSampling} \label{lem:sampling:non-uniformity-BCGN}
Let $x \in \R^n$, $\epS, \deltaS \in (0,1)$.
Let $S \in \R^{m\times n}$ be a scaled sampling matrix with $m = 2n \nu^2 \epS^{-2} \logFrac{1}{\deltaS}$
for some $\nu \in (0,1]$.
Then $S$ is an $(\epS, \deltaS)$-JL embedding for $x$ with probability at least $1-\deltaS$. 
\end{lemma}

\begin{proof}[Proof of Theorem \ref{thm:sampling}]
Let $x \in \R^n$.
Using \autoref{lem:lInftoltwo}, we have that with probability at least $\noisyVectorOverallProb$, \eqref{eq:tmp1} and \eqref{eq:nuxHashingUpper} hold.
Conditioning on the above event, namely, \eqref{eq:tmp1} and \eqref{eq:nuxHashingUpper} hold, \autoref{lem:BCGNSampling} gives that with probability at least $1-\deltaS$, a randomly generated sampling matrix $S\in \R^{m\times n}$ is an $\epsilonS$-JL embedding for $\{\tilde{x}\}$.
Since the randomness in $x$ and the randomness in $S$ are independent, we have that with a probability at least 
$\bracket{\noisyVectorOverallProb}\bracket{1 - \deltaS}$, a randomly generated sampling matrix $S\in \R^{m\times n}$ is an $\epsilonS$-JL embedding for $\{\tilde{x}\}$. Using \eqref{eq:tmp1} gives the embedding result for $\{ x\}$.
\end{proof}

\section{Discussion and conclusion}

We see that the presence of the noise modifies the data in such a way that the non-uniformity is reduced, thus making it more suitable for sparse embeddings such as subsampling and hashing. Although naturally concerns may arise on the question of how much original signal we recover, \autoref{lem:noisyvecNorm} answers this question. We see that the presence of the noise does modify the norm of the origianl signal, but due to properties of a high dimensional Gaussian vector, the modification is highly predictable we approximately recover the square norm of the original signal, after a multiplicative factor $1 + \sigma^2 n$. This means that in the presence of the specific noise considered in this paper, dimensionality reduction that approximately preserves the norm becomes easier than without noise. Why? It turns out that the key is the noise contains information about our target, namely, the norm of the original vector. Therefore in this situation the noise is not only removable, but helpful. \autoref{fig::no_noise} and \autoref{fig::noise} illustrate using subsampling for approximate norm-preserving embeddings without or with noise, when the origianl vector has high non-uniformities ($x \in \R^1000$ with three non-zero entries of equal magnitude.) The contrast is remarkable. 

		\begin{figure}[]
		    \centering
		    \begin{minipage}{0.48\textwidth}
		        \centering
		        \includegraphics[width=\textwidth]{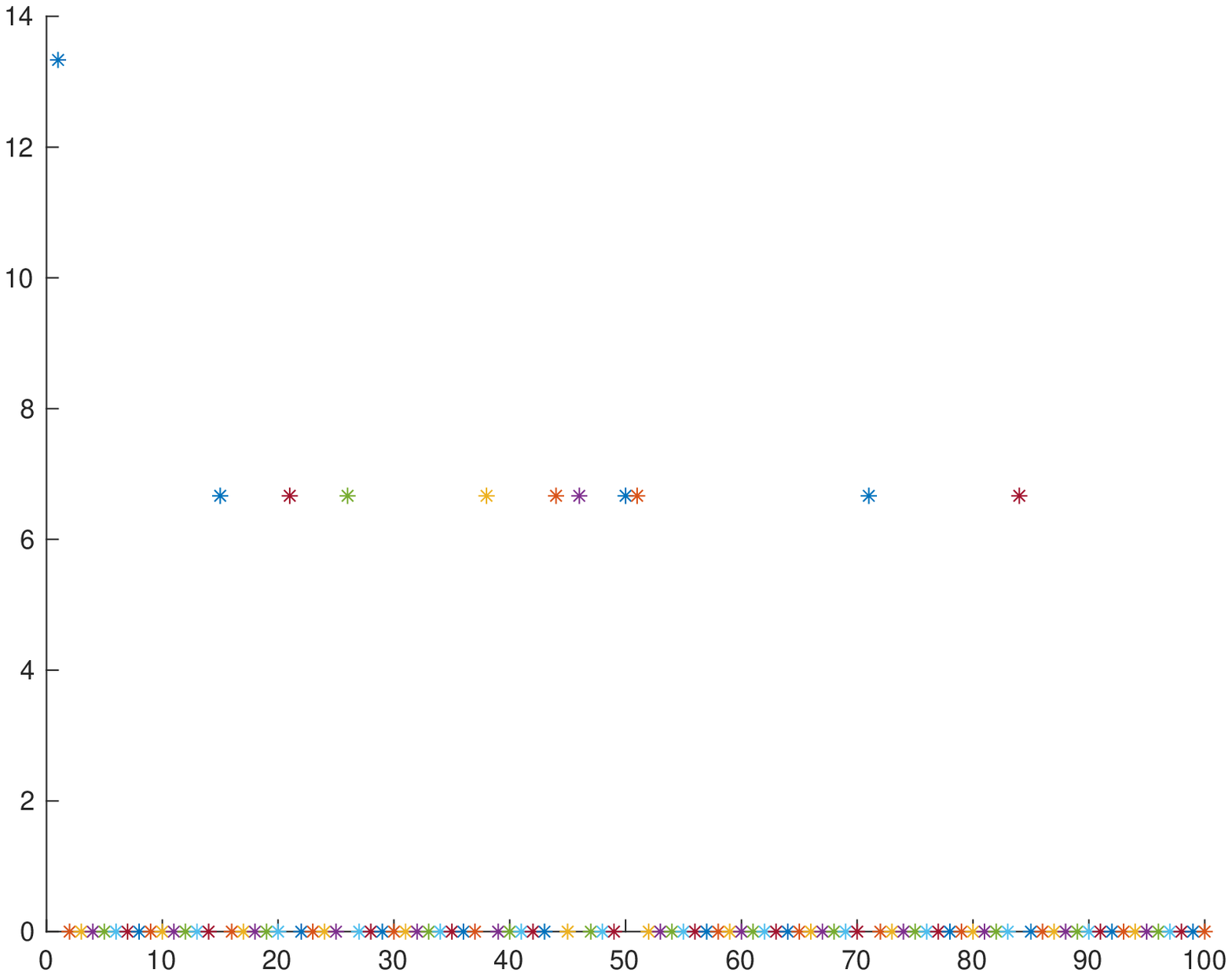} 
		        \caption{Scatter plot of $\normTwo{Sx}^2/\normTwo{x}^2$ over $100$ independent experiments where $S$ is a scaled sampling matrix. We see that without noise, subsampling fails to generate an approximate norm preserving embedding.}
		        \label{fig::no_noise}
		    \end{minipage}\hfill
		    \begin{minipage}{0.48\textwidth}
		        \centering
		        \includegraphics[width=\textwidth]{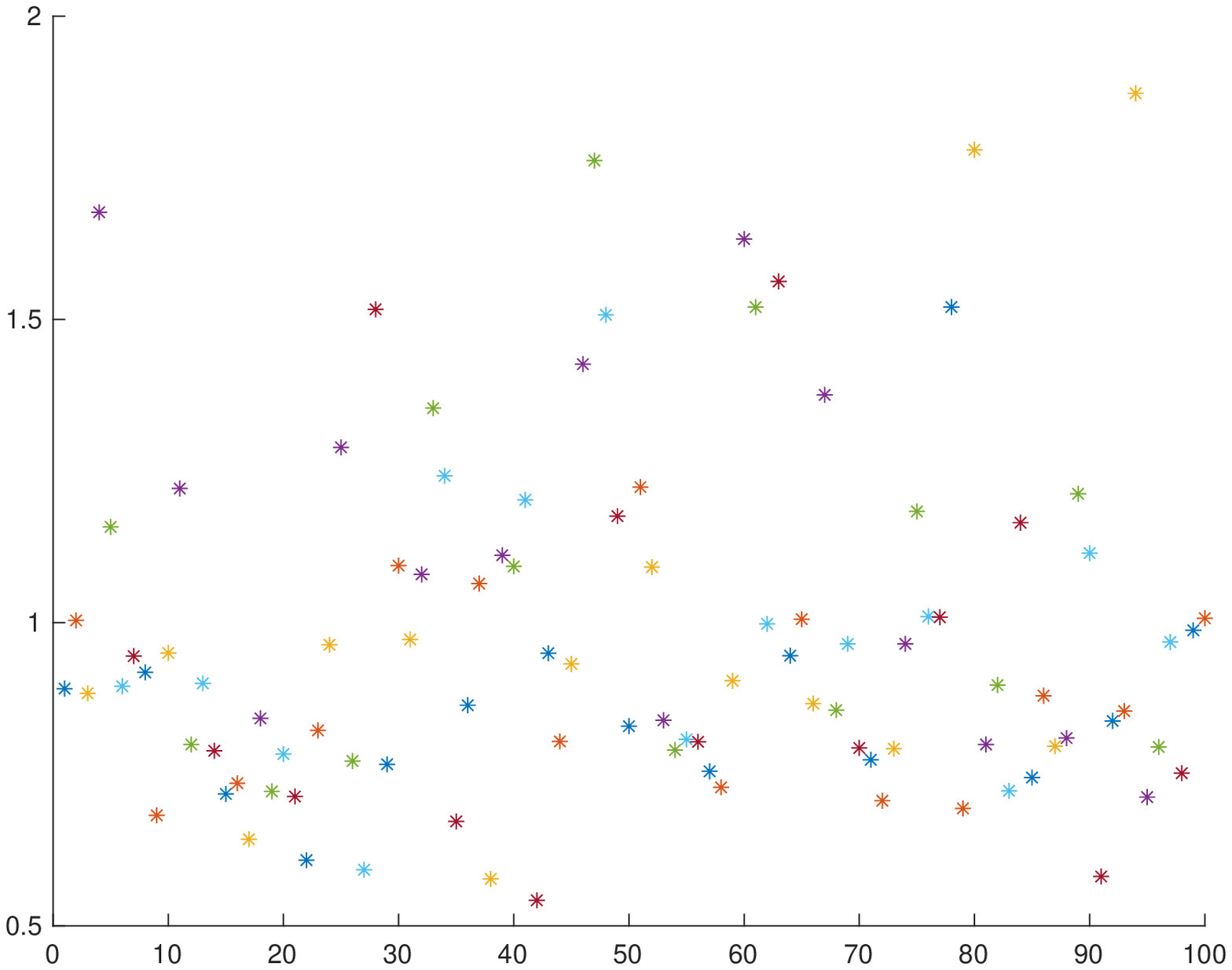} 
		        \caption{Scatter plot of $\normTwo{S\tilde{x}}^2/(\sigma^2 n \normTwo{x}^2)$ over $100$ independent experiments where $S$ is a scaled sampling matrix, $\sigma = 0.1, n=1000$. We see that with noise, subsampling succeeds to generate an approximate norm preserving embedding.}
		        \label{fig::noise}
		    \end{minipage}   
		\end{figure}

\paragraph{Discussion on \autoref{thm:hashing}} \eqref{eqn:n0} is satisfied for $n = \Omega \bracket{\frac{1}{\nuBar \sigma^2}}$, assuming $n = o \bracket{e^{\frac{1}{\sigma^2}}}$. For $n = \mathO{e^{\frac{1}{\sigma^2}}}$, \eqref{eqn:n0} is also satisfied provided that $e^{\frac{1}{\sigma^2}}  = \Omega \bracket{\frac{1}{\nuBar}}$, which holds for small $\sigma$. Therefore with very weak assumptions, $1$-hashing matrices in a noisy data setting have comparable embedding properties with Gaussian matrices, namely, the $\mathO{\epS^{-2} \logFrac{1}{\deltaS}}$ embedding dimension, but superior computational properties for $n = \Omega\bracket{\frac{1}{\nuBar \sigma^2}}$. Note that regardless of relative magnitude of $\sigma$ and $\nuBar$, \eqref{eqn:n0} always hold as $n \to \infty$.

\paragraph{Discussion on \autoref{thm:sampling}} With $n = \mathO{\frac{1}{\sigma^2}}$, we have from \autoref{thm:sampling} that \\ $m = \mathO{\epS^{-2} \logFrac{1}{\deltaS} \max \bracket{\frac{\nu(x)^2}{\sigma^2}, \log(n)} }$, the first term in the maximum is proportional to the required embedding dimension of sampling matrices for data without noise (see \autoref{lem:BCGNSampling}). As $n \to \infty$, we have $m = \mathO{\epS^{-2} \logFrac{1}{\deltaS} \log(n)}$. Thus for data in sufficiently high dimensions, subsampling requires an additional $\log(n)$ factor in terms of the embedding dimension compared with Gaussian matrices. The $\log(n)$ factor is expected due to the coupon collector's problem. 

Further questions remain. What are the embedding results for other noisy models, for example, Rademacher instead of Gaussian noises, or sparse noises? The author thinks, however, that the analysis in this paper could be easily extended to cover all the sub-Gaussian type of noise. Another direction that the author intends to take is to analyze subspace-embeddings, where not only one/a finite set of vectors are projected with norms approximately preserved, but vectors in an entire column space of a given matrix are projected with norms approximately preserved. The author already has key proofs for subspace-embedding type results and is currently writing it up for publication. 


In conclusion, in this paper, we propose to consider $JL$-embeddings in a setting where the input is corrupted by a Gaussian noise. We analysed $1$-hashing matrices and subsampling techniques in such a setting, showing that both achieve comparable embedding dimensions with Gaussian matrices when the data dimension is sufficiently high. Such results could be useful to understand performances of embeddings in applications such as in nearest neighbour search, or subspace methods for non-convex optimisatinons. Most importantly, we illustrate that noise is not always to be feared, in specific situations one takes advantage of the noise to achieve one's objective.

\bibliography{Reference.bib}        
\bibliographystyle{abbrv}  

\appendix

\section{Numerical experiment on non-uniformity of a noisy vector}
To illustrate that the non-uniformity ($l_{\infty}$ to $l_2$ ratio) of a noisy vector converges to $\mathO{\sqrt{\frac{\log(2n)}{n}}}$ for large $n$, we include a simple MATLAB experiment. The vector $x \in \R^n$ is set to have $1$ in the first entry, and zero otherwise. Such an $x$ has the maximum non-uniformity of $1$. We set $\sigma = 0.01$ and $\tilde{x} = x + \sigma \normTwo{x} r$ where $r$ is $N(0,1)$ as before. We let $n$ be 20 equally spaced grid points in the interval $10^4$ to $10^8$, and \autoref{fig:1} shows the loglog plot of $\nu(\tilde{x})$ versus $\sqrt{\frac{\log(n)}{n}}$. We see that the plot is close to a straightline, and a linear regression indicates that the gradient of the line is very close to 1 (1.02), indicating that $\nu(\tilde{x}$ is approximately proportional to $\sqrt{\frac{\log(n)}{n}}$.

\begin{figure}
    \centering
    \includegraphics[width=\textwidth]{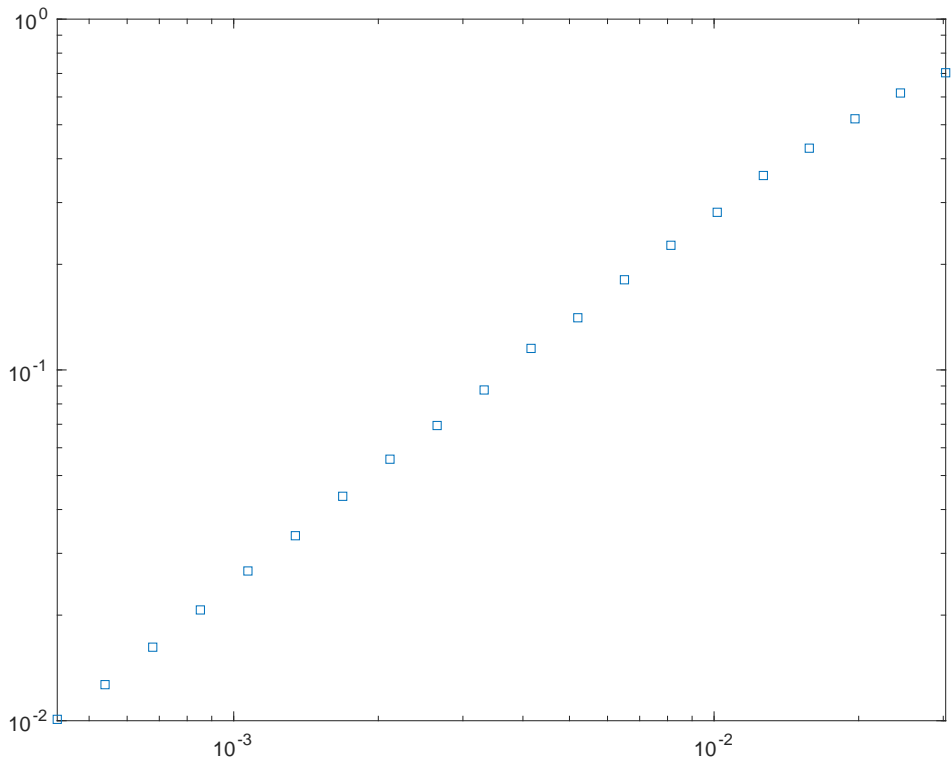}
    \caption{Log-log plot of $\nu(\tilde{x})$ versus $\sqrt{\frac{\log(n)}{n}}$. The data approximately fits a straghtline with unit gradient, indicating a proportionality relationship between the variables.}
    \label{fig:1}
\end{figure}

\end{document}

%% file: preamblesAndMacros.tex
\usepackage{silence}
    \usepackage[utf8]{inputenc}
    \usepackage[T1]{fontenc}
    \usepackage{siunitx}
    \usepackage{colonequals}
    \usepackage[super]{nth}
    \usepackage{algorithm}
    \usepackage{amsmath}
    \usepackage{amsthm}
    \usepackage{amssymb}
    \usepackage{tikz}
    \usepackage{url}
    \usepackage{float}
    \usepackage{tabularx}
    \usepackage{pbox}
    \usepackage{tabulary}
    \usepackage{breqn}
    \usepackage{longtable}
    \usepackage{mdframed}
    \usepackage{booktabs}

    \usepackage{graphicx}
    \graphicspath{{./images/}{./}}

    \usepackage{wrapfig}
    \usepackage[font={small}]{caption}
    \usepackage{enumitem}
    \usepackage{physics}
    \usepackage{lipsum}
    \usepackage{mwe}
    \usepackage{breqn}
    \usepackage[noend]{algpseudocode}
    \usepackage{subcaption}
    \usepackage{graphicx}
    \usepackage{enumitem}
    \usepackage{bm}
    \usepackage{textcmds}

    \usepackage{apptools}
    \AtAppendix{\counterwithin{lemma}{section}}
    \AtAppendix{\counterwithin{remark}{section}}
    
    \usepackage{hyperref}
    \hypersetup{
        colorlinks,
            citecolor=black,
        filecolor=black,
        linkcolor=black,
        urlcolor=black
    }

\newtheorem{theorem}{Theorem}
\newtheorem{lemma}{Lemma}
\newtheorem{remark}{Remark}

\newtheorem{definition}{Definition}

\numberwithin{equation}{section} 
\numberwithin{lemma}{section} 
\numberwithin{theorem}{section} 
\numberwithin{definition}{section} 
\numberwithin{corollary}{section} 

\newcommand{\noisyVectorOverallProb}{1 - 2e^{-\frac{n\epsilon^2}{12}} - 2e^{-t^2/2} - \bracket{\frac{1}{2n}}^{C_1 -1}}
\newcommand{\noisyVectorError}{\frac{\epsilon \sigma^2 n + 2 \sigma t}{1 + \sigma^2 n}}
\newcommand{\epsilonS}{\epsilon_S}
\newcommand{\nuBar}{\bar{\nu}}
\newcommand{\logFrac}[2]{\log \bracket{ \frac{#1}{#2}}}

\newcommand{\mathO}[1]{\mathcal{O}\left( #1\right)}

\newcommand{\N}{\mathbb{N}}
\newcommand{\R}{\mathbb{R}}

\renewcommand{\P}{\mathbb{P}}

\newcommand{\normTwo}[1]{\left\lVert#1\right\rVert_2}

\newcommand{\innerProduct}[2]{\langle #1, #2 \rangle}

\newcommand{\epS}{\epsilon_S}

\renewcommand{\abs}[1]{\lvert #1\rvert}

\newcommand{\nuOneHashingSCTwo}{C_2 \sqrt{\epS} \min \left\{ \frac{\log(E/\epS)}{\log(1/\deltaS)}, \sqrt{\frac{\log(E)}{\log(1/\deltaS)}} \right\}}

\newcommand{\probability}[1]{\P \left( #1 \right)}

\newcommand{\bracket}[1]{\left( #1\right)}

\newcommand{\normInf}[1]{\|#1\|_{\infty}}

\newcommand{\deltaS}{\delta_S}

